\newif\ifprocs
\procstrue
\procsfalse 

\newif\ifapps
\appsfalse

\ifprocs
\documentclass[oribibl,envcountsame]{llncs}
\else
\documentclass[11pt,letterpaper]{article}
\usepackage{fullpage}
\usepackage{setspace}
\fi

\usepackage{amsmath,amssymb,amsfonts}
\usepackage{xspace}

\usepackage{comment}

\usepackage[lined,boxed,linesnumbered]{algorithm2e}

\ifprocs
\else
  \usepackage{theorem}

\newtheorem{fact}{Fact}[section]
\newtheorem{lemma}[fact]{Lemma}
\newtheorem{theorem}[fact]{Theorem}
\newtheorem{definition}[fact]{Definition}

\newenvironment{proof}{{\bf Proof:  }}{\hfill\rule{2mm}{2mm}}
\fi

\newcommand*{\QEDA}{\hfill\ensuremath{\blacksquare}}

\newcommand{\symd}{\triangle}

\newcommand\polylog[1]{\ensuremath{\mathrm{polylog}\left(#1\right)}}

\allowdisplaybreaks

\ifprocs \else \newcounter{note}[section] \fi

\usepackage{url}

\usepackage{ifpdf}
\ifpdf    
\usepackage{hyperref}
\else    
\usepackage[hypertex]{hyperref}
\fi

\title{Explicit Expanding Expanders}
\ifprocs
\author{Michael Dinitz\inst{1}\fnmsep\thanks{Supported in part by NSF grant \#1464239}
\and Michael Schapira\inst{2}\fnmsep\thanks{Supported in part by ISF grant 420/12, Israel Ministry of Science Grant 3-9772, Marie Curie Career Integration Grant, the Israeli Center for Research Excellence in Algorithms (I-CORE), Microsoft Research PhD Scholarship} 
\and Asaf Valadarsky\inst{3}\fnmsep$^{\star\star}$
\institute{Johns Hopkins University.
\email{mdinitz@cs.jhu.edu}
\and Hebrew University of Jerusalem. 
\email{schapiram@huji.ac.il}
\and Hebrew University of Jerusalem. 
\email{asaf.valadarsky@mail.huji.ac.il}
} }
\else
\date{}
\author{Michael Dinitz\thanks{Supported in part by NSF grant \#1464239}\\Johns Hopkins University\\ \url{mdinitz@cs.jhu.edu} 
\and 
Michael Schapira\thanks{Supported in part by ISF grant 420/12, Israel Ministry of Science Grant 3-9772, Marie Curie Career Integration Grant, the Israeli Center for Research Excellence in Algorithms (I-CORE), and a Microsoft Research PhD Scholarship} \\Hebrew University of Jerusalem\\ \url{schapiram@huji.ac.il} 
\and 
Asaf Valadarsky\footnotemark[2]\\Hebrew University of Jerusalem\\ \url{asaf.valadarsky@mail.huji.ac.il}}
\fi

\begin{document}

\maketitle

\begin{abstract}
Deterministic constructions of expander graphs have been an important topic of research in computer science and mathematics, with many well-studied constructions of infinite families of expanders.  In some applications, though, an infinite family is not enough: we need expanders which are ``close" to each other.  We study the following question: Construct an an infinite sequence of expanders $G_0,G_1,\ldots,$ such that for every two consecutive graphs $G_i$ and $G_{i+1}$, $G_{i+1}$ can be obtained from $G_i$ by adding a single vertex and inserting/removing a small number of edges, which we call the \emph{expansion cost} of transitioning from $G_i$ to $G_{i+1}$. This question is very natural, e.g., in the context of datacenter networks, where the vertices represent racks of servers, and the expansion cost captures the amount of rewiring needed when adding another rack to the network. We present an \emph{explicit} construction of $d$-regular expanders with expansion cost at most $\frac{5d}{2}$, for any $d\geq 6$. Our construction leverages the notion of a ``2-lift'' of a graph.  This operation was first analyzed by Bilu and Linial~\cite{BL06}, who repeatedly applied $2$-lifts to construct an infinite family of expanders which double in size from one expander to the next. Our  construction can be viewed as a way to ``interpolate'' between Bilu-Linial expanders with low expansion cost while preserving good edge expansion throughout. 
\ifprocs \else

While our main motivation is centralized (datacenter networks), we also get the best-known distributed expander construction in the ``self-healing" model.
\fi
\end{abstract}

\ifprocs
\else
\thispagestyle{empty}

\clearpage

\setcounter{page}{1}
\fi

\section{Introduction}

Expander graphs (aka expanders) have been the object of extensive study in theoretical computer science and mathematics (see e.g.~the survey of~\cite{HLW06}). Originally introduced in the context of building robust, high-performance communication networks~\cite{BP73}, expanders are both very natural from a purely mathematical perspective and play a key role in a host of other applications (from complexity theory to coding). While $d$-regular random graphs are, in fact, very good expanders~\cite{Bol88,Fri08}, many applications require \emph{explicit}, deterministic constructions of expanders.\footnote{Throughout this paper we will use ``explicit'' and ``deterministic'' interchangeably.} Consequently, a rich body of literature in graph theory deals with deterministic constructions of expanders, of which the best known examples are Margulis's construction~\cite{Mar73} (with Gabber and Galil's analysis~\cite{GG81}), algebraic constructions involving Cayley graphs such as that of Lubotzky, Phillips, and Sarnak~\cite{LPS88}, constructions that utilize the zig-zag product~\cite{RVW00}, and constructions that rely on the concept of $2$-lifts~\cite{BL06,MSS13}.

All of these constructions generate an infinite family of $d$-regular expanders. However, for important applications of expanders that arise in computer networking, this is not enough. Our primary motivating example are datacenters, which network an unprecedented number of computational nodes and are the subject of much recent attention in the networking research community. Consider a datacenter network represented as a graph, in which each vertex represents a rack of servers, and edges represent communication links between these racks (or, more accurately, between the so-called ``top-of-rack switches''). Expanders are natural candidates for datacenter network topologies as they fare well with respect to crucial objectives such as fault-tolerance and throughput~\cite{BP73,HLW06}. However, the number of racks $n$ in a datacenter grows regularly as new equipment is purchased and old equipment is upgraded, calling for an expander construction that can grow gracefully (see discussion of industry experience in~\cite{SHPG12}, and references therein).

We hence seek expander constructions that satisfy an extra constraint: incremental growth, or \emph{expandability}.  When a new rack is added to an existing datacenter, it is impractical to require that the datacenter be entirely rewired and reconfigured. Instead, adding a new rack should entail only a small number of local changes, leaving the vast majority of the network intact. From a theoretical perspective, this boils down to requiring  that the construction of expanders not only work for all $n$, but also involve very few edge insertions and deletions from one expander to the next.

Our aim, then, is to explicitly construct an infinite family of expanders such that (1) every member of the family has good (edge) expansion; and (2) every member of the family can be obtained from the previous member via the addition of a single vertex and only ``a few'' edge insertions and deletions. Can this be accomplished? What are the inherent tradeoffs (e.g., in terms of edge expansion vs.~number of edge insertions/deletions)? We formalize this question and take a first step in this direction. Specifically, we present the first construction of explicit expanding expanders and discuss its strengths and limitations.

\subsection{Our Results and Techniques}

We formally define edge expansion and expansion cost in Section~\ref{sec:prelims}.  We now provide an informal exposition. The \emph{edge expansion} of a set of vertices is the number of edges leaving the set divided by the size of the set, and the edge expansion of a graph is the worst-case edge expansion across all sets. The \emph{expansion cost} for a graph $G_i$ on $n$ vertices $\{1,\ldots,n\}$ and graph $G_{i+1}$ on $n+1$ vertices $\{1,\ldots,n+1\}$ is the number of edge insertions and removals required to transition from $G_i$ to $G_{i+1}$.  The expansion cost of a family of graphs $\{G_i=(V_i,G_i)\}$, where $V_{i+1}$ is the union of $V_i$ and an additional vertex, is the worst-case expansion cost across all consecutive pairs of graphs in the family. Observe that adding a new vertex to a $d$-regular graph while preserving $d$-regularity involves inserting $d$ edges between that vertex and the rest of the graph, and removing at least $\frac{d}{2}$ edges to ``make room'' for the new edges. Hence, $\frac{3d}{2}$ is a lower bound on the expansion cost of any family of $d$-regular graphs.

Our main result is an explicit construction of an infinite family of $d$-regular expanders with very good edge expansion and small expansion cost:
\begin{theorem} \label{thm:overview-main}
For any even degree $d \geq 6$, there exists an infinite sequence of explicitly constructed $d$-regular expanders $\{G_i=(V_i,E_i\}$ such that
\begin{enumerate}
\item $|V_0|=\frac{d}{2} + 1$, and for every $i\geq 0$, $|V_{i+1}|=|V_{i}|+1$.
\item The edge expansion of $G_i$ is at least $\frac{d}{3} - O(\sqrt{d \log^3 d})$  for every $i\geq 0$.
\item The expansion cost of the family $\{G_i\}$ is at most $\frac{5d}{2}$.
\end{enumerate}
\end{theorem}

The attentive reader might notice that we claim our graphs are $d$-regular, yet the number of vertices of the first graph in the sequence, $G_0$, is only $\frac{d}{2}+1$. This seeming contradiction is due to our use of multigraphs, i.e.,~graphs with parallel edges. In particular, $G_0$ is the complete graph on $\frac{d}{2} + 1$ vertices, but where every two vertices are connected by $2$ parallel edges. While expanders are traditionally simple graphs, all nice properties of $d$-regular expanders, including the relationships between edge and spectral expansion, continue to hold with essentially no change for $d$-regular ``expander multigraphs''.

Our construction technique is to first deterministically construct an infinite sequence of ``extremely good'' expanders by starting at $K_{\frac{d}{2} +1}$ and repeatedly ``2-lifting'' the graph~\cite{BL06}. This standard and well-studied approach to explicitly constructing an infinite sequence of expanders was introduced in the seminal work of Bilu and Linial~\cite{BL06}. However, as every 2-lift doubles the size of the graph, this construction can only generate expanders on $n$ vertices where $n = 2^i (\frac{d}{2}+1)$ for some $i \geq 1$.  We show how to ``interpolate'' between these graphs. Intuitively, rather than doubling the number of vertices all at once, we insert new vertices one at a time until reaching the next Bilu-Linial expander in the sequence. Our construction and proof crucially utilize the properties of $2$-lifts, as well as the flexibility afforded to us by using multigraphs.

While our main focus is on centralized constructions for use as datacenter networks, the fact that our construction is deterministic also allows for improved expander constructions in some \emph{distributed} models.  Most notably, we get improved ``self-healing" expanders.  In the self-healing model, nodes are either inserted or removed into the graph one at a time, and the algorithm must send logarithmic-size messages between nodes (in synchronous rounds) in order to recover to an expander upon node insertion or removal.  Clearly small expansion cost is a useful property in this context.  The best-known construction of self-healing expanders~\cite{Dex} gives an expander with edge expansion of at least $d / 20000$, $O(1)$ maximum degree, $O(1)$ topology changes, and $O(\log n)$ recovery time and message complexity (where the time and complexity bounds hold with high probability, while the other bounds hold deterministically).  Our construction gives a self-healing expander with two improvements: much larger edge expansion (approximately $d/6$ rather than $d/20000$), and deterministic complexity bounds.  In particular, we prove the following theorem:

\begin{theorem}
For any $d \geq 6$, there is a self-healing expander which is completely deterministic, has edge expansion at least $d/6 - o(d)$, has maximum degree $d$, has $O(d)$ topology changes, and has recovery time and message complexity of $O(\log n)$.
\end{theorem} 

\subsection{Related Work}

The immediate precursor of this paper is a recent paper of Singla et al.~\cite{SHPG12}, which proposes random graphs as datacenter network topologies. \cite{SHPG12} presents a simple randomized algorithm for constructing a sequence of random regular graphs with small expansion cost. 
While using random graphs as datacenter topologies constitutes an important and thought-provoking experiment, the inherent unstructuredness of random graphs  poses obstacles to their adoption in practice.  Our aim, in contrast, is to \emph{explicitly} construct expanders with \emph{provable} guarantees on edge expansion and expansion cost.

The deterministic/explicit construction of expanders is a prominent research area in both mathematics and computer science. See the survey of Hoory, Linial, and Wigderson~\cite{HLW06}. Our approach relies on the seminal paper of Bilu and Linial~\cite{BL06}, which proposed and studied the notion of $2$-lifting a graph. They proved that when starting with any ``good'' expander, a random $2$-lift results in another good expander and, moreover, that this can be derandomized. Thus \cite{BL06} provides a means to deterministically construct an infinite sequence of expanders: start with a good expander and repeatedly 2-lift. All expanders in this sequence are proven to be quasi-Ramanujan graphs, and are conjectured to be Ramanujan graphs (i.e., have optimal spectral expansion). Marcus, Spielman, and Srivastava~\cite{MSS13} recently showed that this is indeed essentially true for \emph{bipartite} expanders.  

There has been significant work on using expanders in peer-to-peer networks and in distributed computing.  See, in particular, the continuous-discrete approach of Naor and Wieder~\cite{NW07}, and the self-healing expanders of~\cite{Dex}.  The main focus of this line of research is on the efficient design of \emph{distributed} systems, and so the goal is to minimize metrics like the number of messages between computational nodes, or the time required for nodes to join/leave the system.  Moreover, the actual degree does not matter (since edges are logical rather than physical links), as long as it is constant.  Our focus, in contrast, is on centralized constructions that work for any fixed degree $d$.

\section{Preliminaries: Expander Graphs and Expansion Cost} \label{sec:prelims}

We adopt most of our notation from the survey of Hoory, Linial, and Wigderson on expanders~\cite{HLW06}.  Throughout this paper the graphs considered are multigraphs without self-loops, that is, may have parallel edges between any two vertices. We will commonly treat a multigraph as a weighted simple graph, in which the weight of each edge is an integer that specifies the number of parallel edges between the appropriate two vertices. Given such a weighted graph $G = (V,E,w)$, let $n = |V|$ and say that $G$ is $d$-regular if every vertex in $V$ has weighted degree $d$.  We let $N(u) = \{v \in V : \{u,v\} \in E\}$ be the neighborhood of vertex $u$ for any vertex $u \in V$. Traditionally, expanders are defined as simple graphs, but it is straightforward to see that all standard results on expanders used here continue to hold for multigraphs.

\vspace{0.05in}\noindent{\bf Expansion:} For $S, T \subseteq V$, let $E(S,T)$ denote the multiset of edges with one endpoint in $S$ and one endpoint in $T$, and let $\bar S = V \setminus S$.  If $G = (V,E,w)$ is a $d$-regular multigraph, then for every set $S \subseteq V$ with $1 \leq |S| \leq \frac{n}{2}$ the \emph{edge expansion} (referred to simply as the \emph{expansion}) of $S$ is
$h_G(S) = \frac{|E(S, \bar S)|}{|S|}$.
We will sometimes omit the subscript when $G$ is clear from context.  The edge expansion of $G$ is 
$h(G) = \min_{S \subseteq V : 1 \leq|S| \leq \frac{n}{2}} h_G(S)$.
We say that $G$ is an expander if $h(G)$ is large.  In particular, we want $h(G)$ to be at least $d/c$ for some constant $c$.

While much of our analysis is combinatorial, we also make extensive use of spectral analysis.  Given a multigraph $G$, the adjacency matrix of $G$ is an $n \times n$ matrix $A(G)$ in which the entry $A_{ij}$ specifies the number of edges between vertex $i$ and vertex $j$.  We let $\lambda_1(G) \geq \lambda_2(G) \geq \dots \geq \lambda_n(G)$ denote the eigenvalues of $A(G)$, and let $\lambda(G) = \max\{\lambda_2(G), |\lambda_n(G)|\}$. 

Cheeger's inequality (the discrete version) enables us relate the eigenvalues of a (multi)graph $G$ to the edge expansion of $G$:
\begin{theorem} \label{thm:cheeger}
$\frac{d-\lambda_2}{2} \leq h(G) \leq \sqrt{2 d (d - \lambda_2)}$.
\end{theorem}

We will also use the \emph{Expander Mixing Lemma}, which, informally, states that the number of edges between any two sets of vertices is very close to the expected number of edges between such sets in a random graph.
\begin{theorem}[\cite{AC89}] \label{thm:mixing}
$\left| |E(S,T)| - \frac{d |S| |T|}{n} \right| \leq \lambda \sqrt{|S| |T|}$ for all $S, T \subseteq V$.
\end{theorem}

\noindent{\bf Bilu-Linial:} The construction of $d$-regular expanders using ``lifts", due to Bilu and Linial~\cite{BL06}, plays a key role in our construction.  Informally, a graph $H$ is called a \emph{$k$-lift} of a (simple) graph $G$ if every vertex in $G$ is replaced by $k$ vertices in $H$, and every edge in $G$ is replaced with a perfect matching between the two sets of vertices in $H$ that represent the endpoints of that edge in $G$. To put this formally: a graph $H$ is called a \emph{$k$-lift} of graph $G$ if there is a function $\pi : V(H) \rightarrow V(G)$ such that the following two properties hold.  First, $|\pi^{-1}(u)| = k$ for all $u \in V(G)$.  Second, if $\{u,v\} \in E(G)$ then for every $x \in \pi^{-1}(u)$ there is exactly one $y \in \pi^{-1}(v)$ such that $\{x,y\} \in E(H)$.

We call the function $\pi$ the \emph{assignment function} for $H$.  We follow Bilu and Linial in only being concerned with $2$-lifts. Observe that if $H$ is a $2$-lift of $G$ then $|V(H)| = 2|V(G)|$ and $|E(H)| = 2|E(G)|$, and furthermore that if $G$ is $d$-regular then so is $H$.  Bilu and Linial proved that when starting out with a $d$-regular expander $G$ that also satisfies a certain sparsity condition (see Corollary 3.1 in~\cite{BL06}), one can deterministically and efficiently find a $2$-lift $H$ where $\lambda(H) \leq O(\sqrt{d \log^3 d})$ and moreover $H$ continues to satisfy the sparsity condition.   As $K_{d+1}$ (the $d$-regular complete graph on $d+1$ vertices) satisfies the sparsity condition, starting out with $K_{d+1}$ and repeatedly $2$-lifting generates a deterministic sequence of $d$-regular expanders, each of which twice as large as the previous, with edge expansion at least $\frac{d - O(\sqrt{d \log^3 d})}{2}$ throughout (see also Theorem 6.12 in~\cite{HLW06}).

\vspace{0.05in}\noindent{\bf Incremental Expansion:} We will also be concerned with the expansion cost of an infinite family of expander (multi)graphs.  Given two sets $A,B$, let $A \symd B = (A \setminus B) \cup (B \setminus A)$ denote their symmetric difference.  Let $\mathcal G = G_1, G_2, \dots$ be an infinite family of $d$-regular expanders, where $V(G_i) \subset V(G_{i+1})$ for all $i \geq 1$.
\begin{definition}
The \emph{expansion cost} of $\mathcal G$ is $\alpha(\mathcal G) = \max_{i \geq 1} |E(G_i) \symd E(G_{i+1})|$.
\end{definition}

As our focus is on multigraphs, the edge sets are in fact multisets, and so the expansion cost is the change in weight from $G_i$ to $G_{i+1}$.  Slightly more formally, if we let $x_i^e$ denote the number of copies of edge $e$ in $E(G_i)$, we have that $\alpha(\mathcal G) = \max_{i\geq 1} \sum_{e \in E(G_{i+1}) \cup E(G_i)} |x_e^i - x_e^{i+1}|$.  Observe that the expansion cost is defined for any infinite sequence of graphs, and that a large gap in size from one graph to the next trivially implies a large expansion cost. We restrict our attention henceforth to constructions that generate a $d$-regular graph on $n$ vertices for every integer $n$. We observe that the expansion cost of any such sequence is at least $\frac{3d}{2}$, since $E(G_{i+1}) \setminus E(G_i)$ must contain $d$ edges incident to the vertex in $V(G_{i+1}) \setminus V(G_i)$, and in order to maintain $d$-regularity there must be at least $\frac{d}{2}$ edges in $E(G_i) \setminus E(G_{i+1})$.

\section{Construction and Some Observations} \label{sec:construction}

We now formally present our construction of the sequence $\mathcal G$ of $d$-regular expanders and prove some simple properties of this construction.  

We begin with the complete graph on $\frac{d}{2} + 1$ vertices and assign every edge a weight of $2$.  This will serve as the first graph in $\mathcal G$. To simplify exposition, we will refer to this graph as $G_{\frac{d}{2} + 1}$. In general, the subscript $i$ in graph $G_i\in \mathcal G$ will henceforth refer to the number of vertices in $G_i$. Clearly, $G_{\frac{d}{2} + 1}$ is $d$-regular and has edge expansion $\frac{d}{2}$.  We now embed the Bilu-Linial sequence of graphs starting from $G_{\frac{d}{2} + 1}$ in $\mathcal G$: for every $i \geq 0$, let $G_{2^{i+1} (\frac{d}{2} + 1)}$ be the $2$-lift of $G_{2^i (\frac{d}{2} + 1)}$ guaranteed by~\cite{BL06} to have $\lambda(G_{2^i (\frac{d}{2} + 1)}) \leq O(\sqrt{d \log^3 d})$ (recall that the next graph in the sequence can be constructed in polynomial time).  Assign weight $2$ to every edge in this sequence of expanders. We refer to graphs in this subsequence of $\mathcal G$ as \emph{BL expanders}, since they are precisely $d/2$-regular BL expanders in which every edge is doubled.  Thus each BL expander is $d$-regular and by the Cheeger inequality has edge expansion at least $\frac{d}{2} - O(\sqrt{d \log^3 d})$.

We let $G^*_i$ denote $G_{2^i (\frac{d}{2} + 1)}$. We know, from the definition of a $2$-lift, that for each $i$ there exists a function $\pi: V(G^*_{i+1}) \rightarrow V(G^*_i)$ which is surjective and has $|\pi^{-1}(u)| = 2$ for all $u \in V(G^*_i)$.  As we want that $V(G^*_i) \subset V(G^*_{i+1})$, we identify one element of $\pi^{-1}(u)$ with $u$, i.e.~for each $u \in V(G^*_i)$ we will assume (without loss of generality) that $u \in V(G^*_{i+1})$ and $\pi(u) = u$.

To construct the infinite sequence $\mathcal G$ it is clearly sufficient to show how to create appropriate expanders for all values of $n$ between $2^i(\frac{d}{2}+1)$ and $2^{i+1}(\frac{d}{2}+1)$ for an arbitrary $i$.  Fix some $i \geq 0$, let $\pi : V(G^*_{i+1}) \rightarrow V(G^*_i)$ be the assignment function for the BL expanders, and initialize the sets $S = \emptyset$ (called the \emph{split} vertices) and $U = V(G^*_i)$ (called the \emph{unsplit} vertices). We apply the following algorithm to construct $G_{n+1}$ from $G_n$, starting with $n = 2^i(\frac{d}{2} + 1)$ and iterating until $n= 2^{i+1}(\frac{d}{2} + 1) - 1$.

\begin{enumerate}

\item {\bf Splitting a vertex $u$ into $u$ and $u'$.} Let $u$ be an arbitrary unsplit vertex.  We let the new vertex in $G_{n+1}$ that is not in $G_n$ be $u'$, the vertex in $\pi^{-1}(u)$ that is not $u$.  Let $S(u) = S \cap N(u)$ be the neighbors of $u$ that have already split, and let $U(u) = U \cap N(u)$ be the neighbors of $u$ that are unsplit. Here the neighborhood $N(u)$ is with respect to $G_n$.  \label{step:node}

\item {\bf Inserting edges from $u$ and $u'$ to unsplit neighbors.}  For every $v \in U(u)$, replace the edge from $u$ to $v$ (which we prove later always exists) with an edge from $u$ to $v$ of weight $1$ and an edge from $u'$ to $v$ of weight $1$. \label{step:unsplit}

\item {\bf Inserting edges from $u$ and $u'$ to split neighbors.} For every pair of vertices $v, v' \in S(u)$ with $\pi(v) = \pi(v')$, decrease the weight of $\{v,v'\}$ by $1$ and do one of the following:\label{step:split}
\begin{itemize}
\item if $\{u,v\} \in E(G^*_{i+1})$, assign $\{u,v\}$ a weight of $2$, remove $\{u,v'\}$, and add an edge $\{u',v'\}$ of weight $2$;
\item otherwise (that is, $\{u, v'\} \in E(G^*_{i+1})$), assign $\{u,v'\}$ a weight of $2$, remove $\{u,v\}$, and add an edge $\{u', v\}$ of weight $2$.
\end{itemize}

\item {\bf Inserting edges between $u$ and $u'$.}  Add an edge between $u$ and $u'$ of weight $|U(u)|$. \label{step:pair}

\item {\bf Mark $u$ and $u'$ as split.} Remove $u$ from $U$, add $u$ and $u'$ to $S$.
\end{enumerate}

We prove the following simple invariants. We will refer to two vertices $u, v$ as \emph{paired} if $\pi(u) = \pi(v)$.  Together, these lemmas imply that the algorithm is well-defined and that we have an infinite sequence of $d$-regular graphs that interpolates between BL expanders.

\begin{lemma} \label{lem:large-weight-invariant}
Let $u, u'$ be paired vertices with $\pi(u) = \pi(u') = u$.  Then throughout the execution of the algorithm, edge $\{u,u'\}$ exists if $u$ has already split and if there are neighbors of $u$ which are unsplit.  If $\{u,u'\}$ exists then it has weight equal to the number of neighbors of $u$ that are unsplit.
\end{lemma}
\ifprocs \else \begin{proof}
When $u$ is first split (when $u'$ is first created) the edge $\{u,u'\}$ has weight $|U(u)|$ by construction.  Now suppose that we are at some point in the execution of the algorithm, let $U(u)$ be the set of neighbors of the original vertex that are still unsplit, and assume that the weight of $\{u,u'\}$ is $|U(u)|$.  We will prove that this invariant continues to hold.  Let $v$ be the vertex that is currently being split, say into $v$ and $v'$.  If $v$ was not a neighbor of $u$ in the original expander then it is not a neighbor of $u$ or $u'$ in the current graph, and clearly splitting it has no effect on the weight of $\{u,u'\}$.  If $v$ was a neighbor of $u$, then when we split $v$ we decrease the weight of $\{u,u'\}$ by $1$.  Observe that now, though, there is one less neighbor of $u$ that is unsplit, and so the invariant is maintained. 
\end{proof}
\fi 

\begin{lemma} \label{lem:edge-weight-invariant}
Edges between unpaired split vertices always have weight $2$, edges between unsplit vertices always have weight $2$, and edges with one endpoint unsplit and one split have weight $1$.
\end{lemma}
\ifprocs \else \begin{proof}
We start out with $G^*_i$ in which no vertices are split and all edges have weight $2$, satisfying the lemma.  Suppose the lemma is satisfied at the moment we split some vertex $u$ into $u$ and $u'$.  Edges between unpaired vertices that do not have $u$ as an endpoint are unchanged.  Edges from $u$ or $u'$ to unsplit vertices have weight $1$ by step~\ref{step:unsplit} of the algorithm, and edges from $u$ or $u'$ to split vertices have weight $2$ by step~\ref{step:split}.  This implies the lemma. 
\end{proof} \fi

\begin{lemma} \label{lem:degree}
Every vertex has weighted degree $d$ throughout the execution of the algorithm.
\end{lemma}
\ifprocs \else \begin{proof}
We proceed by induction.  For the base case, take the original expander $G^*_i$: it is $\frac{d}{2}$ regular and every edge has weight $2$, so the weighted degree is $d$.  Now, suppose that we just split the vertex $u$ into $u$ and $u'$, and assume that before the split $u$ had weighted degree $d$. Lemma~\ref{lem:edge-weight-invariant} implies that before the split each edge from $u$ to a vertex in $S(u)$ had weight $1$, while each edge from $u$ to a vertex in $U(u)$ had weight $2$.  Thus, $|S(u)| + 2|U(u)| = d$.

After the split, the edges from $u$ and from $u'$ to vertices that are unsplit (i.e.~vertices in $U(u)$) have weight $1$, while the edges to vertices in $S(u)$ have weight $2$ (by Lemma~\ref{lem:edge-weight-invariant}).  However, each of $u$ and $u'$ is adjacent to only half of the vertices in $S(u)$, since for each $v, v'$ pair in $S(u)$ the edges $\{u,v\}$ and $\{u, v'\}$ are replaced by the appropriate matching (either $\{u, v\}, \{u',v'\}$ or $\{u,v'\}, \{u', v\}$).  By construction, we know that the weight of $\{u,u'\}$ is $|U(u)|$.  Hence, $u$ and $u'$ have weighted degree $2\frac{|S(u)|}{2} + |U(u)| + |U(u)| = |S(u)| + 2|U(u)| =  d$.

Now, consider some vertex $v \in U(u)$.  By Lemma~\ref{lem:edge-weight-invariant}, before splitting $u$ the edge from $u$ to $v$ had weight $2$.  After splitting, $v$ has a weight $1$ edge to $u$ and a weight $1$ edge to $u'$, and thus maintains its weighted degree of $d$.

Lastly, let $v \in S(u)$, with its paired vertex $v'$.  By Lemma~\ref{lem:edge-weight-invariant}, before splitting $u$ the edge from $u$ to $v$ (and the one to $v'$) had weight $1$.  After splitting, it is replaced by a single edge of weight $2$ (to either $u$ or $u'$, depending on the matching). However, the weight on the $\{v,v'\}$ edge is also decreased by $1$, and so the total weighted degree of $v$ is unchanged (note that Lemma~\ref{lem:large-weight-invariant} implies that since $v \in S(u)$ the weight of $\{v,v'\}$ before splitting $u$ is at least $1$, so this weight can be decreased by $1$ without becoming negative).
\end{proof} \fi

\begin{lemma} \label{lem:correct}
When all vertices have split, $G$ is precisely $G^*_{i+1}$ in which all edges have weight $2$.
\end{lemma}
\ifprocs \else \begin{proof}
We proceed by induction, with the inductive hypothesis that the edges between non-paired split vertices are exactly the edges between those vertices in $G^*_{i+1}$.  Initially there are no split vertices, so this holds.  Now suppose it holds for $G_n$, and suppose we create $G_{n+1}$ by splitting $u$ into $u$ and $u'$.  Then the only changes in the edges between split vertices are the addition of edges from $u$ and $u'$ to vertices in $S(u)$.  But step~\ref{step:split} explicitly creates those edges to be identical to the edges in $G^*_{i+1}$, and thus the inductive hypothesis continues to hold.  This, together with Lemmas~\ref{lem:large-weight-invariant}~and~\ref{lem:edge-weight-invariant}, implies the lemma. 
\end{proof} \fi


\section{Analysis: Expansion and Expansion Cost}\label{sec:analysis}

We next prove that that the expansion cost of our construction is small, and the edge expansion throughout is good. Specifically, we prove that the expansion cost is at most $\frac{5}{2}d$, and then prove some combinatorial lemmas which will immediately imply that the edge expansion is at least $\frac{d}{4} - O(\sqrt{d \log^3 d})$. We show in Section~\ref{sec:improved-expansion} how this bound on edge expansion can be improved to a tight lower bound of $\frac{d}{3} - O(\sqrt{d \log^3 d})$ via a more delicate, spectral analysis combined with the combinatorial lemmas from this section.  

We begin by analyzing the expansion cost.

\begin{theorem} \label{thm:cost}
$\alpha(\mathcal G) \leq \frac52 d$.
\end{theorem}
\begin{proof}
Suppose $G_{n+1}$ is obtained from $G_n$ by splitting vertex $u$ into $u$ and $u'$.  The transition from $G_n$ to $G_{n+1}$ entails the following changes in edge weights:
\begin{itemize}
\item {\bf A change of $2$ in edge weights per vertex in $U(u)$.} Each edge from vertex $u$ to a vertex $v \in U(u)$ changes its weight from $2$ to $1$ and an additional edge of weight $1$ is added from $u'$ to $v$, so there are $2$ edge changes per vertex in $U(u)$.

\item {\bf A change of $5$ in edge weights for every two paired vertices in $S(u)$.} Every pair of edges in $G_n$ (of weight $1$) from $u$ to paired vertices $v, v'$ in $S(u)$ is replaced by a pair of edges between $u, u'$ and $v, v'$, each of weight $2$, which results in a total change in edge weights of 4: 1 for increasing the weight of one of $u$'s outgoing edges to the pair $v,v'$ from $1$ to $2$, 1 for decreasing an edge of $u$'s other outgoing edge from $1$ to $0$, and $2$ for the new edge from $u'$ the pair $v, v'$.  In addition, the weight of the edge $(v,v')$ is decreased by $1$.  So, each pair of vertices in $S(u)$ induces a total change of $5$ in edge weights.

\item {\bf An additional change of $|U(u)|$ in edge weights.} An edge of weight $|U(u)|$ is added between $u$ and $u'$.

\end{itemize}

Hence, $|E(G_n) \symd E(G_{n+1}| = 2|U(u)| + 5|S(u)| / 2 + |U(u)| = 3|U(u)| + (5|S(u)| / 2)$.  As $2 |U(u)| + |S(u)| = d$ by Lemma~\ref{lem:degree}, this concludes the proof of the theorem. 
\end{proof}

This analysis is tight for our algorithm. At some point in the execution of the algorithm, some vertex $u$ will be split after all of its neighboring vertices have already been split. As this entails a change in weight of $5$ for each of the $\frac{d}{2}$ paired vertices in $S(u)$, the resulting total change in edge weights will be $\frac{5}{2}d$.


\subsection{Edge Expansion} \label{sec:expansion}

We show, via a combinatorial argument, that every member of our sequence of graphs $\mathcal G$ has edge expansion at least $\frac{d}{4} - O(\sqrt{d \log^3 d})$. To this end, we show that for every $n$ between $2^i(\frac{d}{2}+1)$ and $2^{i+1}(\frac{d}{2}+1)$, the graph $G = G_n = (V,E)$ has edge expansion at least $\frac{d}{4}  - O(\sqrt{d \log^3 d})$. We will then show in Section~\ref{sec:improved-expansion} how this lower bound on edge expansion can be tightened to $\frac{d}{3} - O(\sqrt{d \log^3 d})$ via spectral analysis combined with the combinatorial lemmas proved here.

\begin{theorem} \label{thm:expansion}
For every $G\in\mathcal G$, $h(G) \geq \frac{d}{4} - O(\sqrt{d \log^3 d})$.
\end{theorem}

We now prove Theorem~\ref{thm:expansion}. Let $S \subseteq V$ denote the set of vertices that have already split in $G$, and let $U \subseteq V$ be the set of vertices that are currently unsplit.  Let $H = (V_H, E_H) = G^*_{i+1}$ be the next BL expander in the sequence and let $\pi$ be its assignment function (note that the range of $\pi$ is the vertices of the previous BL expander, which includes the vertices $U$ in $G$).  For any subset $A \subseteq V$, let $F(A) \subseteq V_H$ denote the ``future'' set of $A$, in which all unsplit vertices in $A$ are split and both vertices appear in $F(A)$.  More formally, $F(A) = (A \cap S) \cup (\cup_{u \in A \cap U} \pi^{-1}(u))$.  For $X, Y \subseteq V_H$ with $X \cap Y = \emptyset$, let $w_H(X,Y)$ denote the total edge weight between $X$ and $Y$ in $H$.  Lastly, for $A, B \subseteq V$ with $A \cap B = \emptyset$ we define $w_G(A,B)$ similarly, except that we \emph{do not} include edge weights between paired vertices. Our proof proceeds by analyzing $w_G(A,B)$ for all possible different subsets of vertices $A,B$ in $G$. As $w_G(A,B)$ only reflects the edge weights in $G$ between non-paired vertices, the proof below lower bounds the actual edge expansion (which also includes weights between paired vertices).

\begin{lemma} \label{lem:SS}
If $A, B \subseteq S$ with $A \cap B = \emptyset$, then $w_H(F(A), F(B)) = w_G(A, B)$.
\end{lemma}
\ifprocs \else \begin{proof}
Since $A, B \subseteq S$, we know by definition that $F(A) = A$ and $F(B) = B$.  This means that (if we ignore edges between $u_0, u_1$ with $\pi(u_0) = \pi(u_1)$) the edges in $G$ between $A$ and $B$ are precisely the edges in $H$ between $A$ and $B$, and moreover all such edges have weight $2$ in both $G$ and $H$.
\end{proof} \fi

\begin{lemma} \label{lem:UU}
If $A, B \subseteq U$ with $A \cap B = \emptyset$, then $w_H(F(A), F(B)) = 2 \cdot w_G(A, B)$.
\end{lemma}
\ifprocs \else \begin{proof}
Since $A$ and $B$ are entirely unsplit, by definition $F(A) = \cup_{u \in A} \pi^{-1}(u)$ and $F(B) = \cup_{u \in B} \pi^{-1}(u)$.  This means that if $a \in A$ and $b \in B$, there is an edge between $a$ and $b$ in $G$ if and only if there is a matching between $\pi^{-1}(a)$ and $\pi^{-1}(b)$ in $H$. Clearly, any such edge $\{a,b\}$ has weight $2$ in $G$ (since neither endpoint has split), and the two edges in the matching between $\pi^{-1}(a)$ and $\pi^{-1}(b)$ in $H$ also have weight $2$ (by definition). Hence, $w_H(F(A), F(B)) = 2 \cdot w_G(A, B)$. 
\end{proof} \fi

\begin{lemma} \label{lem:SU}
If $A \subseteq S$ and $B \subseteq U$, then $w_H(F(A), F(B)) = 2 \cdot w_G(A, B)$.
\end{lemma}
\ifprocs \else \begin{proof}
Clearly $F(A) = A$ and $F(B) = \cup_{u \in B} \pi^{-1}(u)$.  Consider an edge $\{a, b\} \in E$ with $a \in A$ and $b \in B$. By Lemma~\ref{lem:edge-weight-invariant}, this edge has weight $1$.  Let $\{b_0, b_1\} = \pi^{-1}(b)$.  Then in $H$ exactly one of $\{a, b_0\}$ and $\{a, b_1\}$ exists, and this edge has weight $2$.  Thus $w_H(F(A), F(B)) \geq 2 \cdot w_G(A,B)$.  Similarly, let $\{a,b\} \in E_H$ with $a \in F(A)$ and $b \in F(B)$.  Then this edge has weight $2$, and in $G$ the edge $\{a, \pi(b)\}$ must exist and have weight $1$.  Hence $w_H(F(A), F(B)) \leq 2 \cdot w_G(A,B)$. \ifprocs \qed \fi
\end{proof} \fi

Combining these lemmas proves that every cut in $G$ has weight at least half of that of the associated ``future" cut, since we can divide any cut in $G$ into split and unsplit parts.  

\begin{lemma} \label{lem:half}
If $(A, \bar A)$ is a cut in $G$, then $w_G(A, \bar A) \geq \frac12 w_H(F(A), F(\bar A))$.
\end{lemma}
\ifprocs \else \begin{proof}
We divide each of $A$ and $\bar A$ into two parts:  let $S(A)$ denote the nodes in $A \cap S$, let $U(A) = A \cap U$, let $S(\bar A) = \bar A \cap S$, and let $U(\bar A) = \bar A \cap U$.  We then have that
\begin{align*}
w_G(A, \bar A) &= w_G(S(A), S(\bar A)) + w_G(S(A), U(\bar A)) + w_G(U(A), S(\bar A)) + w_G(U(A), U(\bar A)) \\
& = w_H(F(S(A)), F(S(\bar A))) + \frac12 w_H(F(S(A)), F(U(\bar A))) \\
&\qquad + \frac12 w_H(F(U(A)), F(S(\bar A))) + \frac12 w_H(F(U(A)), F(U(\bar A))) \\
& \geq \frac12 w_H(F(A), F(\bar A))
\end{align*}
where the first equality is by definition (since $S$ and $U$ are disjoint) and the second equality is due to Lemmas~\ref{lem:SS}, \ref{lem:UU}, and~\ref{lem:SU}. The last inequality is again because $F(S(A)), F(U(A)), F(S(\bar A))$, and $F(U(\bar A))$ are disjoint. 
\end{proof} \fi

\ifprocs \else
Let $X \subseteq V$ with $|X| \leq |\bar X|$.  We know that in $H$ the edge expansion of $X$ is at least $d/2 - O(\sqrt{d \log^3 d})$, and so
\begin{equation*}
h_G(X) = \frac{w_G(X, \bar X)}{|X|} \geq \frac{\frac12 w_H(F(X), \overline{F(X)})}{\min\{|F(X)|, |\overline{F(X)}|\}}  = \frac12 h_H(F(X)) \geq \frac{d}{4} - O\left(\sqrt{d \log^3 d}\right).
\end{equation*}

Theorem~\ref{thm:expansion} follows.
\fi

\section{Improved Edge Expansion Analysis}\label{sec:improved-expansion}

We proved in Section~\ref{sec:expansion} that our sequence of graphs has edge expansion at least $\frac{d}{4} - O(\sqrt{d \log^3 d})$. We next apply spectral analysis to improve this lower bound. 

\begin{theorem} \label{thm:improved-expansion}
For every $G\in\mathcal G$, $h(G) \geq \frac{d}{3} - O(\sqrt{d \log^3 d})$.
\end{theorem}

Interestingly, while we prove this theorem by using spectral properties of Bilu-Linial expanders, we cannot prove such a theorem through a direct spectral analysis of the expanders that we generate.  

\begin{theorem} \label{thm:rayleigh}
For any $\epsilon > 0$, there are an infinite number of graphs $G \in \mathcal G$ which have $\lambda_2(G) \geq d/2 - \epsilon$. 
\end{theorem}
\ifprocs \else \begin{proof}
Fix $i \geq 0$, and let $G_{n-1} = G^*_i$.  Let $G_n$ be the next graph in $\mathcal G$, obtained by splitting a single node of $G^*_i$ (say $v$) into two nodes (say $v_0$ and $v_1$).  So the weight of the edge between $v_0$ and $v_1$ in $G_n$ is $d/2$.  Recall that $\lambda_1(G_n) = d$ and the associated eigenvector is the vector $\bf{1/\sqrt{n}}$ in which every coordinate is $1/\sqrt{n}$.  So in order to lower bound $\lambda_2(G_n)$, we just need to find a vector $\vec{x}$ orthogonal to $\bf{1/\sqrt{n}}$ with Rayleigh quotient $(\vec{x}^T A \vec{x}) / (\vec{x}^T \vec{x}) \geq d/2 - \epsilon$.   

Let $\vec x$ be the vector with $1-2/n$ in the coordinate for $v_0$ and $1-2/n$ in the coordinate for $v_1$, and $-2/n$ in all other coordinates.  Then clearly $\vec x$ is orthogonal to $\bf{1/\sqrt{n}}$.  We begin by analyzing $\vec{x}^T A \vec{x} = \sum_i \sum_j A_{ij} x_i x_j$.  Simple calculations show that when $i$ is not in the neighborhood of $v$ it contributes $\Theta(d/n^2)$ to this sum, while if $i$ is in the neighborhood of $v$ then it contributes $\Theta(d/n^2 - 1/n) = -\Theta(1/n)$.  Finally, if $i$ is $v_0$ or $v_1$ then it contributes $\frac{d}{2}(1 - \frac2n)^2 - (\frac{d}{2} - 1)(\frac2n)(1-\frac{2}{n})$.  Thus 
\begin{equation*}
\vec{x}^T A \vec{x} \geq d \left(\frac{n-2}{n}\right)^2 - \Theta(d/n) \geq d - \epsilon
\end{equation*}
for large enough $n$. 

Now we are left with the easy task of computing $\vec{x}^T \vec{x}$.  This is clearly $2(\frac{n-2}{n})^2  + (n-2)(4/n^2) \leq 2+\epsilon$ for large enough $n$.  Putting this together, we get that the Rayleigh quotient of $x$ is at least $(d-\epsilon) / (2+\epsilon) \geq d/2 - \epsilon$ (for large enough $n$, by slightly changing $\epsilon$).  Thus $\lambda_2(G_n) \geq d/2 - \epsilon$.  This was true for all sufficiently large $n$, so by setting $i$ large enough we have this infinitely often.  
\end{proof} \fi

This implies that if we want to lower bound $h(G)$ by using Theorem~\ref{thm:cheeger} (the Cheeger inequalities), the best bound we could prove would be $d/4$.  Thus Theorem~\ref{thm:improved-expansion} beats the eigenvalue bound for this graph.

We now begin our proof of Theorem~\ref{thm:improved-expansion}.  We use the same terminology and notation as in the proof of Theorem~\ref{thm:expansion}. The key to improving our analysis lies in leveraging the fact that $H= G^*_{i+1}$, the next BL expander in the sequence of graphs $\mathcal G$, is a strong \emph{spectral} expander (i.e., $\lambda(G^*_{i+1}) \leq O(\sqrt{d \log^3 d})$).  We first handle the case of unbalanced cuts, then the more difficult case of nearly-balanced cuts.  We then show that the analysis in this section is tight.

\vspace{0.1in}\noindent{\bf Unbalanced Cuts.} We first show that in a strong spectral expander, unbalanced cuts give large expansion.  This is straightforward from the Mixing Lemma (Theorem~\ref{thm:mixing}) if the cut is not \emph{too} unbalanced, i.e.~if both sides of the cut are of linear size. However, a straightforward application of the Mixing Lemma fails when the small side is very small.  
We show that this can be overcome by using the full power of the Mixing Lemma: the two sets in Theorem~\ref{thm:mixing} need not be a cut, but can be any two sets.

\begin{lemma} \label{lem:unbalanced-general}
If $X \subseteq V_H$ with $|X| \leq n/2$, then $w_H(X, \bar X) \geq |X| \left(d \left(\frac{n-|X|}{n}\right) - 4\lambda\right)$.
\end{lemma}
\ifprocs \else \begin{proof}
Recall that $H$ is $\frac{d}{2}$-regular and all edges have weight $2$.  Consider any bisection $(X_0, X_1)$ of $X$ such that $X_0 \cap X_1 = \emptyset$ and $|X_0| = |X_1| = |X|/2$.  The Mixing Lemma implies that $|E(X_0, X_1)| \leq \frac{(d/2) \cdot (|X|^2/4)}{n} + \lambda \frac{|X|}{2}$. We claim that this implies that the number of edges with both endpoints in $X$ is at most $\frac{d |X|^2}{4n} + \lambda |X|$.  To see this, suppose otherwise.   Then in a random bisection of $X$ (i.e., a random partition of $X$ into two equally-sized subsets) the expected number of edges across the bisection is larger than $\frac{d}{2} \cdot \frac{|X|^2}{4n} + \lambda \frac{|X|}{2}$.  Hence there exists a bisection of $X$ with at least that many edges across it, contradicting our upper bound on the number of edges across any bisection.

So the total number of edges with both endpoints in $X$ is at most $\frac{d |X|^2}{4n} + \lambda |X|$.  Each of these edges counts against the total degree for two vertices (each endpoint), and so
$|E(X, \bar X)| \geq \frac{d}{2} |X| - \frac{d|X|^2}{2n} - 2 \lambda |X| = |X| \left(\frac{d}{2} \cdot \frac{n - |X|}{n} - 2 \lambda\right)$.

This, and the fact that every edge has weight $2$, concludes the proof
\end{proof} \fi

\begin{lemma} \label{lem:unbalanced}
If $X \subseteq V$ with $|X| < \frac{n}{5}$, then $h_G(X) \geq \frac{d}{3} - O\left(\sqrt{d \log^3 d}\right)$.
\end{lemma}
\ifprocs \else \begin{proof}
Clearly $|F(X)| \leq 2|X|$ and $|F(\bar X)| \geq |\bar X|$.  Thus, $|F(X)| < \frac13 |V_H|$, and so Lemma~\ref{lem:unbalanced-general} implies that $w_H(F(X), F(\bar X)) \geq \left(\frac{2d}{3} - O\left(\sqrt{d \log^3 d}\right)\right) |F(X)|$.  Now, Lemma~\ref{lem:half} and the fact that $|F(X)| \geq |X|$ imply that $w_G(X, \bar X) \geq \left(\frac{d}{3} - O\left(\sqrt{d \log^3 d}\right)\right) |X|$, giving the claimed expansion. 
\end{proof} \fi

\vspace{0.1in}\noindent{\bf Balanced Cuts.} We next prove that $h_G(X) \geq \frac{d}{3} - O(\sqrt{d \log^3 d})$ when $\frac{n}{5} \leq |X| \leq \frac{n}{2}$.  To accomplish this, we use the Mixing Lemma (again) to show that the expansion does not drop by a factor of $2$ from the future cut.  Intuitively, if $X$ contains many unsplit vertices, then even though $G$ only gets half of the weight from unsplit vertices than $H$ does, there are only half as many vertices and thus the expansion is basically preserved.\footnote{We point out that this is not quite accurate, since $F(X)$ could be larger than $F(\bar X)$.}  On the other hand, if $X$ contains many split vertices, then either $\bar X$ also contains many split vertices (and so by Lemma~\ref{lem:SS} we lose nothing), or $\bar X$ contains many unsplit vertices  (and so the cut is unbalanced enough for the Mixing Lemma to provide stronger bounds).

\begin{lemma} \label{lem:balanced}
If $X \subseteq V$ with $\frac{n}{5} \leq |X| \leq \frac{n}{2}$, then $h_G(X) \geq \frac{d}{3} - O\left(\sqrt{d \log^3 d}\right)$.
\end{lemma}
\begin{proof}
As before, let $S(X) = S \cap X, U(X) = U \cap X, S(\bar X) = S \cap \bar X$, and $U(\bar X) = U \cap \bar X$.  We first analyze the weight of the future cut using the Mixing Lemma (Theorem~\ref{thm:mixing}).
\begin{align}
\label{eq:weight-split} w_H(&F(X), F(\bar X)) = w_H(F(S(X)), F(S(\bar X))) + w_H(F(S(X)), F(U(\bar X))) \\
\nonumber &\quad + w_H(F(U(X)), F(S(\bar X))) + w_H(F(U(X)), F(U(\bar X))) \\
\label{eq:big-mixing} &\geq \frac{d  |F(S(X))| \cdot |F(S(\bar X))|}{|F(X)| + |F(\bar X)|}+ \frac{d  |F(S(X))| \cdot |F(U(\bar X))|}{|F(X)| + |F(\bar X)|}  \\
\nonumber & \quad + \frac{d \cdot |F(U(X))| \cdot |F(S(\bar X))|}{|F(X)| + |F(\bar X)|}+ \frac{d \cdot |F(U(X))| \cdot |F(U(\bar X))|}{|F(X)| + |F(\bar X)|} -4\lambda |V_H|\\
\label{eq:resize} & \geq d  \frac{|S(X)| (|S(\bar X)| + 2|U(\bar X)|) + 2 |U(X)| (|S(\bar X)| + 2 |U(\bar X)|)}{|X| + |\bar X| + |U(X)| + |U(\bar X)|}  - 4\lambda |V_H|.
\end{align}
Equation~\eqref{eq:weight-split} is simply the partition of the edges crossing the cut into the natural four sets.  Equation~\eqref{eq:big-mixing} is the application of the Mixing Lemma to each of the four parts, together with an upper bound of $|V_H|$ on all sets to bound the discrepancy due to the Mixing Lemma to $4 \lambda |V_H|$.  Equation~\eqref{eq:resize} exploits the fact that unsplit vertices in $V$ split into exactly two vertices in $V_H$ to get that $|V_H| = |F(X)| + |F(\bar X)| = |X| + |\bar X| + |U(X)| + |U(\bar X)|$, and that $|F(S(X))| = |S(X)|$, $|F(S(\bar X))| = |S(\bar X)|$, $|F(U(X))| = 2 |U(X)|$, and $|F(U(\bar X))| = 2 |U(\bar X)|$.  

We can now apply Lemmas~\ref{lem:SS}, \ref{lem:UU}, and \ref{lem:SU} to relate this to the weight in $G$.  The first term in~\eqref{eq:resize} remains unchanged, whereas the second, third, and fourth terms are reduced by a factor of $2$, and the final loss term also remains unchanged.  With these adjustments, we get that
\begin{align*}
w_G(X, \bar X) &\geq \frac{d\left(|S(X)| \left(|S(\bar X)| + |U(\bar X)|\right) + |U(X)| \left(|S(\bar X)| + 2 |U(\bar X)|\right)\right)}{|X| + |\bar X| + |U(X)| + |U(\bar X)|} - 4\lambda |V_H| \\
&= d \cdot \frac{|S(X)| \cdot |\bar X| + |U(X)| \cdot \left(|\bar X| + |U(\bar X)|\right)}{|X| + |\bar X| + |U(X)| + |U(\bar X)|} - 4\lambda |V_H| \\
&= d \cdot \frac{|X| \cdot |\bar X| + |U(X)| \cdot |U(\bar X)|}{|X| + |\bar X| + |U(X)| + |U(\bar X)|} - 4\lambda |V_H|.
\end{align*}

Note that $\lambda$ in this expression is $\lambda(H)$, not $\lambda(G)$.  We can now get the expansion simply by dividing by $|X|$, the size of the smaller side:
$h_G(X) \geq d \cdot \frac{|X| \cdot |\bar X| + |U(X)| \cdot |U(\bar X)|}{|X| \left(|X| + |\bar X| + |U(X)| + |U(\bar X)|\right)} - 40 \lambda$,
where for the final term we use the fact that $|V_H| \leq 2n$ and $|X| \geq \frac{n}{5}$ to get that $4 \lambda |V_H| / |X| \leq \lambda \cdot 8n / (\frac{n}{5}) = 40 \lambda$.

We claim that this expression is at least $\frac{d}{3}- O(\sqrt{d \log^3 d})$.  As $\lambda = O(\sqrt{d \log^3 d})$, it needs to be shown that $\frac{|X| \cdot |\bar X| + |U(X)| \cdot |U(\bar X)|}{|X| \left(|X| + |\bar X| + |U(X)| + |U(\bar X)|\right)} \geq \frac{1}{3}$.  Suppose for the sake of contradiction that this is false.  Then rearranging terms gives us that
\begin{equation} \label{eq:end}
|U(X)| \cdot (3 |U(\bar X)| - |X|) < |X|^2 - 2 |X| |\bar X| + |X| |U(\bar X)|.
\end{equation}

If $|U(\bar X)| > \frac{|X|}{3}$, then \eqref{eq:end} implies that
$|U(X)| < |X|^2 - 2 |X| |\bar X| + |X| |U(\bar X)| \leq |X|^2 - |X| |\bar X| \leq 0$,
where we used the fact that $|U(\bar X)| \leq |\bar X|$ and $|\bar X| \geq |X|$.  This is a contradiction, since $|U(X)|$ clearly cannot be negative.

Otherwise, if $|U(\bar X)| \leq \frac{|X|}{3}$, then~\eqref{eq:end} implies that
\begin{align*}
|U(X)| >\frac{2|X| |\bar X| - |X|^2 - |X| |U(\bar X)|}{|X| - 3 |U(\bar X)|} \geq \frac{|X|^2 - |X| |U(\bar X)|}{|X| - 3 |U(\bar X)|} \geq |X|,
\end{align*}
since $|\bar X| \geq |X|$.  This is also a contradiction, as $U(X) \subseteq X$, and hence the lemma follows.  
\end{proof}

Combining Lemma~\ref{lem:unbalanced} and Lemma~\ref{lem:balanced} concludes the proof of Theorem~\ref{thm:improved-expansion}.

\vspace{0.1in}\noindent{\bf Tightness of Analysis.} We show that the bound on the edge expansion from Theorem~\ref{thm:improved-expansion} is essentially tight and, moreover, is tight infinitely often.  

\begin{theorem} \label{thm:tight}
There exists a graph in $\mathcal G$ with edge expansion at most $\frac{d}{3} + \frac{2}{3}$ and, for every $i \geq 1$, there exists a graph in $\mathcal G$ between $G^*_i$ and $G^*_{i+1}$ with edge expansion at most $\frac{d}{3} + O(\sqrt{d \log^3 d})$.
\end{theorem}
\ifprocs \else \begin{proof}
Recall that the starting point of the construction of $\mathcal G$ was the graph $K_{\frac{d}{2} + 1}$ with weight $2$ on all edges. After inserting $\frac13 (\frac{d}{2} + 1)$ new vertices, the resulting graph $G$ is on $\frac{4}{3}(\frac{d}{2} + 1)$ vertices.  Consider the cut $(S, U)$ in $G$, where $S$ is all of the vertices that have been split and $U$ is all of the unsplit vertices.  Then $|S| = |U| = \frac{2}{3} (\frac{d}{2} + 1)$, and there is an edge of weight $1$ from every vertex in $S$ to every vertex in $U$.  Consequently, the expansion of $S$ is equal to $\frac{2}{3} (\frac{d}{2} + 1)$.

Similarly, suppose that the initial graph is $G^*_i$ on $n = 2^i(\frac{d}{2}+1)$ vertices and $\frac{n}{3}$ new vertices are inserted to get graph $G$. Consider the cut in $G$ with all the split vertices $S$ on one side and all of the unsplit vertices $U$ on the other.  This cut is a bisection, where each side has size $\frac{2n}{3}$.  In the associated future cut $(F(S), F(U))$ of $G^*_{i+1}$, $|F(S)| = \frac{2n}{3}$ and $|F(U)| = \frac{4n}{3}$.  A simple application of the Mixing Lemma establishes that the weight in $G^*_{i+1}$ across this future cut is at most $\frac{4}{9}nd + O(n \sqrt{d \log^3 d})$, and then Lemma~\ref{lem:SU} implies that the weight across $(S,U)$ in $G$ is at most $\frac29 nd + O(n \sqrt{d \log^3 d})$.  Hence, $h_G(S) \leq \frac{d}{3} + O(\sqrt{d \log^3 d})$. 
\end{proof} \fi

\ifprocs \else
\section{Self-Healing Expanders} \label{app:healing}
We will now show how our construction can be used to build self-healing expanders.  The self-healing model is a variant of the well-known $\mathcal{CONGEST}$ model for distributed computing.  In the $\mathcal{CONGEST}$ model,  we think of the current graph $G = (V, E)$ as the communication graph of a distributed system (in particular, as a peer-to-peer or overlay network).  Each node has a unique id (possibly set by an adversary) which can be used to identify it.  Time passes in synchronous rounds, and in each round every node can send an $O(\log n)$-bit message on each edge incident to it (possibly a different message on different edges), as well as receive a message on each edge.  Usually the complexity of algorithms in this model is given as bounds on the \emph{round complexity} (the number of rounds necessary for the algorithm to complete) and on the \emph{message complexity} (the total number of messages sent during the algorithm).  Local computation is free, since the focus is on the cost of communication.

A \emph{self-healing expander} (originally defined by~\cite{Dex}) is an algorithm in the $\mathcal{CONGEST}$ model which maintains an expander upon node insertions and deletions.  Slightly more formally, given a current graph $G$, the adversary can add a new node or delete a node.  If a node is added, the adversary connects it to a constant-sized subset of current nodes.  If a node is deleted, its neighbors are informed.  This results in an \emph{intermediate graph} $U$.  The recovery algorithm then needs to recover to an expander by changing edges (or adding or deleting edges).  Adding an edge between two nodes $u$ and $v$ requires sending a message from $u$ to $v$ (or vice versa).  Initially, a newly inserted node only knows its (adversarially chosen) id, and does not have any knowledge of the graph.  

The key assumption is that the adversary does not interfere during recovery: no more nodes fail or are deleted until recovery is complete.  However, the adversary is fully-adaptive -- it knows the entire state and all previous states, as well as the details of the algorithm.  

The important parameters of a self-healing expander are 1) the expansion of the graph, 2) the maximum degree, 3) the number of topology changes (i.e.~the expansion cost), 4) the recovery time (i.e.~the round complexity), and 5) the message complexity.  The current best bounds on this are due to Pandurangan, Robinson, and Trehan, who gave a construction they called DEX of a self-healing expander~\cite{Dex} with maximum degree $O(1)$, only $O(1)$ topology changes, and $O(\log n)$ recovery time and message complexity.  We can use our deterministic expander construction to get similar bounds, but with two improvements: much larger edge expansion, and deterministic (rather than high probability) complexity bounds.  

In particular, DEX is based on the ``$p$-cycle with chords", a well-known $3$-regular graph with $\lambda_2 \leq 3(1-\frac{1}{10^4})$ (see, e.g.,~\cite[Section~11.1.2]{HLW06}).  Hence the edge expansion guaranteed by the Cheeger inequality is $\frac{d}{20000} = \frac{3}{20000}$.  Since our construction is based on $2$-lifts, we end up getting expansion $d/6 - o(d)$.  Also, while DEX is an expander with probability $1$, the logarithmic complexity bounds are only with high probability.  Since our expander construction is entirely deterministic, the complexity bounds are also deterministic.  Putting everything together, we get the following theorem.  

\begin{theorem} \label{thm:healing}
For any $d \geq 6$, there is a self-healing expander which is completely deterministic, has edge expansion at least $d/6 - o(d)$, has maximum degree $d$, has $O(d)$ topology changes, and has recovery time and message complexity of $O(\log n)$.
\end{theorem}

\subsection{Algorithm}

At a high level, we will simply maintain the (unweighted) version of our expander construction.  Since our analysis of expansion did not use any edge with weight greater than $2$, using the unweighted version yields a graph with degrees between $d/2$ and $d$ in which the expansion is at least $d/6 - o(n)$.  We just need to show how to maintain this in a distributed manner when a node is inserted or deleted.  Our major advantage over previous approaches (e.g., \cite{Dex}) is that since our expander construction is deterministic, if a node knows the total number of nodes $n$ in the network then it knows the actual topology of the network (since we do not charge for local memory use or computation).  Of course, we cannot simply distribute the value of $n$ throughout the network as that would take too many messages, but it turns out the structure of our expander makes it possible to estimate $n$ well enough to do recovery.  

We will heavily use the concept of a ``name".  Unlike an adversarially assigned ID, a name corresponds to an exact location in the graph.  Names will evolve over time, but intuitively they correspond to the ``splitting history".  We can define names in BL expanders inductively.  In the $i$th BL expander $G^*_i$, the names will be the elements of the set $\{0, 1, \dots, d/2\} \times \{0,1\}^i$.  Recall that $G^*_0$ is a $((d/2)+1)$-clique denoted by $G_0^*$, and hence we can assign unique names by using an arbitrary bijection between the nodes and $\{0,1, \dots, d/2\}$.  To define names in the BL expander $G^*_i$, let $u \in V(G^*_{i-1})$ be an arbitrary node in the previous BL expander and let  $\{u, u'\} = \pi^{-1}(u) \subset V(G^*_i)$ be the two nodes that $u$ has split into in $G^*_i$.  Then the name of $u$ in $G^*_i$ will be the name of $u$ in $G^*_{i-1}$ together with an extra coordinate equal to $0$, and the the name of $u'$ in $G^*_i$ will be the name of $u$ in $G^*_{i-1}$ together with an extra coordinate equal to $1$.  We will let the \emph{length} of a name be the number of bits after the leading element from $\{0,1, \dots, d/2\}$, so, e.g., an element of $\{0,1,\dots, d/2\} \times \{0,1\}^i$ has length $i$.

Let $G_n$ be one of our explicit expanders, with $2^i (\frac{d}{2}+1) \leq n \leq 2^{i+1} (\frac{d}{2} + 1)$.  Then we can define names in the obvious way.  If $u \in V(G_n)$ has not been split, then the name of $u$ is equal to its name in $G^*_i$.  If $u$ has been split, then its name is equal to its name in $G^*_{i+1}$.  So the names of split nodes have one bit more than the names of unsplit nodes.  

We begin by proving a simple lemma: if our graph is $G_n$ and every node knows its name and the names of its neighbors, then we can route messages.  Note that we do \emph{not} assume that every node knows $n$.

In the rest of this section, we will assume a unique shortest path between every two nodes.  If more than one shortest path exists, then we can pick one arbitrarily (it does not matter how we break ties, so long as we are consistent).

\begin{lemma} \label{lem:routing}
Let $G = G_n$, and suppose that every node in $G$ knows its name and the names and ids of its neighbors.  Then any node $u$ can send a message to any other node $v$ along a shortest path in $G$, as long as $u$ knows the name of $v$.
\end{lemma}
\begin{proof}
Suppose that the name of $u$ has length $i$.  Note that while $u$ does not know $n$, the length of its name implies that $G$ is either between $G^*_{i-1}$ and $G^*_i$ or between $G^*_{i}$ and $G^*_{i+1}$.  If at least one neighbor of $u$ has a name of a different length, then this resolves the ambiguity (although $u$ still does not know $n$ precisely), but it might be the case that all neighbors of $u$ have names of the same length.  

By induction, we simply need to show that $u$ can forward the message on the next hop of the shortest path to $v$ in $G$.  If no neighbor of $u$ has a name that is longer than the name of $u$ (i.e.~they all have length $i$ or $i-1$), then $u$ calculates the next hop $w'$ on a shortest path to $v$ in $G^*_i$.  If $w'$ is a neighbor of $u$ then we set $w = w'$.  If $w'$ is not a neighbor of $u$ then this must be because $\pi(w')$ has not yet split and $\pi(w')$ is a neighbor of $u$, in which case we set $w = \pi(w')$.  We then send the message to $w$ (note that this is possible since $G^*_i$ is deterministic and so it does not take any communication for $u$ to know the topology of $G^*_i$).  

If a neighbor of $u$ has a name of length $i+1$, then $u$ pretends that it is the contraction of the two nodes in $\pi^{-1}(u)$ in $G^*_{i+1}$ and calculates the shortest path to $v$ in $G^*_{i+1}$.  More formally, $let G^*_{i+1} / \pi^{-1}(u)$ denote the graph obtained by contracting the two nodes of $\pi^{-1}(u)$ in $G^*_{i+1}$, and let $u'$ denote this contracted node.  Let $w$ denote the next hop on the shortest path from $u'$ to $v$ in $G^*_{i+1} / \pi^{-1}(u)$.  Then either $w$ or $\pi(w)$ is a neighbor of $u$ in $G$, and it is straightforward to see that in either case, it is the next hop on the shortest path from $u$ to $v$ in $G$.  So $u$ can forward the message correctly.
\ifprocs \qed \fi
\end{proof}

We can now define the recovery algorithm for insertions and deletions.  Throughout, we will refer to the node with name $\vec 0$ as the \emph{coordinator} node.  We will assume that the state of $\vec 0$ is always replicated at every neighbor of $\vec 0$: this can be done using an additional $O(d) = O(1)$ messages whenever $\vec 0$ or a neighbor of $\vec 0$ is updated.  

\paragraph{Insertions:} Suppose that a new node $u$ is inserted, adjacent to some arbitrary constant-size subset of current nodes.  Let $v$ be an arbitrary initial neighbor of $u$.  
\begin{enumerate}
\item $u$ sends a message to $v$, asking it to send a message to $\vec 0$ notifying $\vec 0$ of the addition of $u$.
\item $\vec 0$ sends $v$ (who forwards to $u$) a message containing the total number of nodes $n$ (including $u$).
\item Since our sequence of expanders is deterministic, $u$ knows the expander $G_n$, and knows which node $x$ is supposed to split into $x, x'$ in order to create $G_n$ from $G_{n-1}$.  $u$ will become $x'$, setting its name accordingly.
\item $u$ sends a message to $x$ (through $v$) notifying $x$ that $u$ will become $x'$ and that $x$ should update its own name (by adding on a $0$).
\item $x$ responds with a list of its neighbors (names and ids).  
\item $u$ and $x$ each send messages to the appropriate neighbors (as defined by $G_n$) to create or delete edges (and inform them of the new names for $x$ and $x'$).
\end{enumerate}

\begin{theorem} \label{thm:insertions}
If before the insertion $G = G_{n-1}$ (the expander in our construction with $n-1$ nodes), then after the insertion recovery algorithm is complete $G = G_n$.  The total number of rounds and messages are both $O(\log n)$.
\end{theorem}
\begin{proof}
First, note that by Lemma~\ref{lem:routing} the algorithm can indeed send the messages it needs to send.  Initially $v$ knows $\vec{0}$ (up to one bit, which it is easy to see does not matter) and hence can send it the original message from $u$.  By induction $\vec{0}$ knows the true value of $n$, so it can update this value and send back to $v$ who can then forward it to $u$ (we can assume that $\vec 0$ knows the name of $v$ since $v$ can simply include it in the message it forwarded from $u$).  Once $u$ knows $n$ it knows the name of $x$ since $G_n$ is a fixed, deterministic graph, and hence can send a message to $x$.  Similarly, $x$ can send a message to $u$ (through $v$) with the names of its neighbors.  Then using Lemma~\ref{lem:routing} again, $u$ can send messages to these neighbors to build exactly the edges that it (now as $x'$) would have in $G_n$.  Hence after the algorithm finished, $G = G_n$.

The complexity bounds are straightforward.  Each step which requires sending a message sends only $O(\log n)$ bits, so these can indeed fit inside of a message (or $O(1)$ messages).  Since we always route on shortest paths and $G$ is an expander, each message traverses at most $O(\log n)$ edges.  Hence the number of messages and the number of rounds are both $O(\log n)$.
\ifprocs \qed \fi
\end{proof}

\paragraph{Deletions:} Suppose that $u$ is deleted from $G = G_n$, and its neighbors are notified.  Let $x$ be the new node added to $G_{n-1}$ to make $G_n$, i.e.~$\{x\} = V(G_n) \setminus V(G_{n-1})$.  
\begin{enumerate}
\item If $u \neq \vec 0$:
	\begin{itemize}
	\item Each neighbor of $u$ checks whether it is on the shortest path from $u$ to $\vec 0$ (note that these neighbors do not necessarily know $n$, but since they know the length of their names they can determine whether they are on the current shortest path from $u$ to $\vec 0$ for the same reasons as in the proof of Lemma~\ref{lem:routing}).  Let $v$ be the neighbor on the shortest path.
	\item $v$ sends a message to $\vec 0$ informing it that $u$ has been deleted, and $\vec 0$ updates $n$ and sets $w = \vec 0$.
	\end{itemize}
\item If $u = \vec 0$, then each neighbor of $u$ determines if it is on the shortest path from $\vec 0$ to $x$ in $G_n$.  Let $w$ be the neighbor of $\vec 0$ on this shortest path.  We will refer to $w$ as the ``temporary coordinator". 
\item Note that at this point, $w$ knows $u$ and also knows $x$ (since $w$ is either the coordinator or the temporary coordinator, which was replicating the state of the coordinator).
\item $w$ sends a message to $x$, telling it to take over for $u$
\item $x$ sends messages to ``undo" its insertion, i.e., we ``unsplit" the node that was split to add $x$.
\item $x$ sends messages to create edges to nodes that were neighbors of $u$.  
\item If $u = \vec 0$, then $x$ gets the state of $u$ prior to its deletion from one of the neighbors of $u$ and decreases $n$ by $1$.
\item $x$ switches its name to the name of $u$.
\end{enumerate}

\begin{theorem} \label{thm:deletions}
If before the deletion $G = G_n$, then after the deletion recovery algorithm is complete $G = G_{n-1}$.  The total number of rounds and messages are both $O(\log n)$.
\end{theorem} 
\begin{proof}
It is obvious by construction that the algorithm results in $G_{n-1}$, since the unique node $x \in V(G_n) \setminus V(G_{n-1})$ is removed and then precisely takes the place of the deleted node.  Hence we need only to prove the complexity bounds.

As with insertions, note that by Lemma~\ref{lem:routing} all messages can be sent to their desired destinations along a shortest path.  Since $G$ is an expander it has diameter $O(\log n)$, so we need to prove that only $O(1)$ messages are sent across the network.  Clearly sending the original message to $\vec 0$ and the message from $w$ to $x$ take only $O(1)$ messages.  Since there are only $5d/2$ edge changes from $G_{n-1}$ to $G_n$, only $O(d) = O(1)$ messages are needed for $x$ to undo its insertion.  Then only another $O(d) = O(1)$ messages are necessary to replace $u$.  Putting this together with the diameter bound, we get message complexity and round complexity of $O(\log n)$ as claimed.  
\ifprocs \qed \fi
\end{proof}

These theorems together will now let us prove Theorem~\ref{thm:healing}.
\paragraph{Proof of Theorem~\ref{thm:healing}.} The round complexity and message complexity bounds are directly implied by Theorems~\ref{thm:insertions} and \ref{thm:deletions}.  The maximum degree, topology change, and expansion bounds follow from our main expander construction by changing each multigraph into a simple graph (reducing all nonzero weights to $1$).  \ifprocs \qed \fi

\fi

\section{Open Questions}\label{sec:open-questions}

\noindent{\bf Better Bounds.} The obvious open question is proving better bounds for expansion and expansion cost, and exploring the space of tradeoffs between them. 

\ifprocs \else \vspace{0.05in}\noindent{\bf Expanding Simple Expanders.} Our construction crucially utilized the flexibility afforded to us by multigraphs. Can we get comparable results if we restrict all expanders in the sequence to be simple graphs?  It is straightforward to adapt our construction to the simple graph setting by relaxing the regularity condition slightly: simply using our current construction but setting each nonzero weight to $1$ will result in degrees between $d/2$ and $d$, with expansion at least $d/6 - o(d)$.  But can we achieve regularity?
\fi

\vspace{0.05in}\noindent{\bf Expanding Spectral Expanders.} Our construction interpolates between Bilu-Linial (BL) expanders, which are very good spectral expanders ($\lambda \leq O(\sqrt{d \log^3 d})$). But Theorem~\ref{thm:rayleigh} implies that some of the expanders that appear between the BL expanders in the sequence are only weak spectral expanders. Can a sequence of strong spectral expanders (say, with $\lambda \leq O(\sqrt{d} \cdot  \polylog{d})$) be constructed with low expansion cost?


\ifprocs \else 
\section*{Acknowledgements}
We warmly thank Nati Linial for many fruitful discussions. We also thank Yonatan Bilu, Alex Lubotzky, Noga Alon, Doron Puder, and Dan Spielman for helpful conversations.
\fi

\ifprocs
\else
\bibliographystyle{plain}
\bibliography{refs}
\fi

\end{document}